\documentclass[a4paper,twocolumn,accepted=2021-05-27,11pt]{quantumarticle}
\pdfoutput=1
\usepackage[utf8]{inputenc}
\usepackage[T1]{fontenc}
\usepackage[english]{babel} 
\usepackage{graphicx}
\usepackage{amsmath,amssymb,amsthm,bm,mathtools,amsfonts,mathrsfs,bbm,dsfont}
\usepackage{physics}
\usepackage{bbold}
\usepackage{cite}

\newcommand{\id}{\ensuremath{\mathds{1}}}
\renewcommand{\vec}[1]{\boldsymbol{#1}}

\newcommand{\hil}{\mathcal{H}}
\newcommand{\fii}{\varphi}

\newcommand{\R}{\mathbb{R}}
\newcommand{\N}{\mathbb{N}}
\newcommand{\Z}{\mathbb{Z}}

\newtheoremstyle{mystyle}
  {6pt}
  {6pt}
  {\normalfont}
  {0pt}
  {\bf}
  {.}
  { }
  {}

\theoremstyle{mystyle}
\newtheorem{theorem}{Theorem}

\newtheorem{lemma}[theorem]{Lemma}

\newtheorem{observation}[theorem]{Observation}

\begin{document}

\title{Quantum Marginal Problem and Incompatibility}

\author{Erkka Haapasalo}
\affiliation{Centre for Quantum Technologies, National University of Singapore, Science Drive 2, Block S15-03-18, Singapore 117543}
\affiliation{Department of Physics and Center for Field Theory and Particle Physics, Fudan University, Shanghai 200433, China}

\author{Tristan Kraft}
\affiliation{Naturwissenschaftlich-Technische Fakult\"at, Universit\"at Siegen, Walter-Flex-Str. 3, D-57068 Siegen, Germany}

\author{Nikolai Miklin}
\affiliation{Institute of Theoretical Physics and Astrophysics, National Quantum Information Center, Faculty of Mathematics, Physics and Informatics, University of Gdansk, 80-952 Gda\'{n}sk, Poland}

\author{Roope Uola}
\affiliation{D\'{e}partement de Physique Appliqu\'{e}e, Universit\'{e}  de Gen\`{e}ve, CH-1211 Gen\`{e}ve, Switzerland}

\maketitle

\begin{abstract}
One of the basic distinctions between classical and quantum mechanics is the existence of fundamentally incompatible quantities. Such quantities are present on all levels of quantum objects: states, measurements, quantum channels, and even higher order dynamics. In this manuscript, we show that two seemingly different aspects of quantum incompatibility: the quantum marginal problem of states and the incompatibility on the level of quantum channels are in many-to-one correspondence. Importantly, as incompatibility of measurements is a special case of the latter, it also forms an instance of the quantum marginal problem. The generality of the connection is harnessed by solving the marginal problem for Gaussian and Bell diagonal states, as well as for pure states under depolarizing noise. Furthermore, we derive entropic criteria for channel compatibility, and develop a converging hierarchy of semi-definite programs for quantifying the strength of quantum memories.
\end{abstract}

\section{Introduction}
Incompatibility of measurements is one of the most intriguing features of quantum theory~\cite{heinosaari2016invitation}. It refers to the fact that certain measurements cannot be performed simultaneously. Although being an abstract concept on the mathematical level, measurement incompatibility is important also from an experimental perspective. It gives rise to the non-classical behaviour in various quantum protocols such as quantum steering \cite{Quintino2014,Uola2014,uola2015one,KiBuUoPe2017}, contextuality \cite{Xu2019,Tavakoli2019}, tests of macrorealism \cite{Clemente2015,UoViCo2018}, quantum communication \cite{CaHeTo2019,SkSuCa2019,uola2019a,Michal2019,Leo2019,Claudio2019,Mori2019,uola2019b,Buscemi2019}, and non-locality~\cite{fine1982hidden,Wolf2009}. Crucially, any set of incompatible measurements results in genuinely quantum behavior given that one possesses a properly chosen catalyst state. Conversely, compatible measurements give rise to only classical behavior regardless of how the system is prepared.

Being intimately related to the quantum advantage in various tasks, it becomes desirable to search for the fundamental properties of quantum theory that allow for the existence of incompatible measurements. Motivated by this question, we attack the problem with a general approach towards the fact that not all quantum resources can be broadcasted. More precisely, we show that a recently introduced concept of channel incompatibility, which includes measurement incompatibility and no-broadcasting as special cases, is one-to-many connected to an instance of the so-called quantum marginal problem.

The quantum marginal problem asks whether, for a set of reduced states, there is a global state that has these states as its marginals. In general, the quantum marginal problem, which is also known as the $N$-representability problem~\cite{ruskai1969n}, is considered to be one of the most fundamental unsolved problems in quantum chemistry~\cite{national1995mathematical}. In fact, the general problem is proven to be QMA-complete~\cite{Liu2006}. Similar to measurement incompatibility, the marginal problem has strong implications on practical applications. Namely, it is related to the monogamy of entanglement, which refers to the fact that entanglement cannot be shared arbitrarily~\cite{Coffman2000}.

In this manuscript, we focus on an instance of the quantum marginal problem, that is the most relevant from the viewpoint of quantum correlations. This allows us to demonstrate the importance of our results on a broad range of topics, e.g., quantum memories and various relevant marginal problems. More precisely, we develop an efficient method for deciding the strength of quantum memories, based on semidefinite programming (SDP). Moreover, we derive entropic criteria for channel incompatibility, characterize the antidegradability of qubit channels, and solve the marginal problem for the general Gaussian case, for pairs of Bell-diagonal states, and for pairs of general pure states under depolarizing noise.

\section{Quantum marginal problem} Given a collection of $\abs{k}$ states $\qty{\varrho_{J_k}}_{k=1}^{\abs{k}}$ where $J_k\subset I$ and $I=\qty{1,\ldots,n}$, the question is whether there exists a global state $\varrho_I$ on $\mathcal{H}_I$, such that $\varrho_{J_k} = \tr_{I\setminus J_k}[\varrho_I]$ for all $k$. For example, given two bipartite states $\varrho_{AB_1}$ and $\varrho_{AB_2}$ the problem is to find a tripartite state $\varrho_{AB_1B_2}$ compatible with them.

Partial results on the quantum marginal problem are known~\cite{coleman1963structure,klyachko2004quantum,klyachko2006quantum,schilling2017reconstructing}, however, most of them concern non-overlapping marginals, i.e. the case of disjoint sets $J_k$ and concern only pure global states. We concentrate on an instance of the quantum marginal problem, where all given marginals are bipartite and overlap on a single party, namely $I:=\qty{A,B_1,\ldots,B_n}$ and $J_k=\{A,\,B_k\}$. One immediate necessary condition for the existence of a global state is the common marginal on $A$, i.e. $\varrho_A={\rm tr}_{B_k}[\varrho_{AB_k}]$ is the same for all $k$. In our scenario, the global state does not need to be pure.

As a special case we have {\it symmetric extendibility}: A bipartite state $\varrho_{AB}$ is said to have $n$ symmetric extensions if there exists a state $\varrho_{AB_1\ldots B_n}$ such that $\varrho_{AB}=\tr_{I\setminus AB_k}[\varrho_{AB_1\ldots B_n}]$ for all $k$. The set of states possessing $n$ symmetric extensions for all $n\geq 2$ coincides with the set of separable states \cite{Chiribella2011,BaeAcin2006,Stormer69,RaggioWerner89}.

A natural quantifier regarding the marginal problem is the consistent robustness. In the simplest case, one has a pair of bipartite states $\bm{\varrho}:=(\varrho_{AB_1},\varrho_{AB_2})$ sharing a common first marginal $\varrho_A$. The consistent robustness is defined as
\begin{equation}
\label{Eq:MarginalRobustness}
\mathcal{R}_F^c[\bm{\varrho}] = \min\qty{t\geq 0\bigg\vert \frac{\bm{\varrho} + t\bm{\tau}}{1+t}\in F},
\end{equation}
where the optimization is performed over all non-negative real numbers $t$ and pairs of states $\bm{\tau}:=(\tau_{AB_1},\tau_{AB_2})$ having $\varrho_A$ as the first marginal. The set $F$ denotes those pairs of states for which the marginal problem has a solution. Note that we do not restrict the pairs $\bm{\tau}$ to belong to the set $F$. The consistent robustness can be defined analogously for the problem of symmetric extendibility and for sets of three or more states.

\section{Quantum channels and compatibility} Quantum channels, which are completely positive trace preserving maps, describe changes in quantum systems induced by, e.g., measurements or time evolution. These are linear maps on operators taking states of an input system $\hil_A$ into states of the output system $\hil_{B}$; we denote such channels as $\Phi_{A\to B}$.

To introduce channel compatibility, we need three systems denoted as $\hil_A$, $\hil_{B_1}$, and $\hil_{B_2}$. Channels $\Phi_{A\to B_1}$ and $\Phi_{A\to B_2}$ are called {\it compatible} if there exists a broadcasting channel from which they can be obtained as marginals \cite{HeMi2017}. More precisely, compatibility refers to the existence of a channel $\Phi_{A\to B_1 B_2}$ such that
$\Phi_{A\to B_1}(\varrho)={\rm tr}_{B_2}[\Phi_{A\to B_1 B_2}(\varrho)] \text{ and } \Phi_{A\to B_2}(\varrho)={\rm tr}_{B_1}[\Phi_{A\to B_1 B_2}(\varrho)]$ for all input states $\varrho$. This formulation can be directly generalized to sets of channels: Consider the set $I:=\{A,B_1,\ldots,\,B_n\}$ and its subsets $J_k=\{A,\,B_k\}$, each associated with a channel $\Phi_{J_k}:=\Phi_{A\to B_k}$. These channels are compatible if there is a channel $\Phi_I$ with the input system $A$ and the multipartite output Hilbert space $\hil_{B_1}\otimes\cdots\otimes\hil_{B_n}$ such that $\Phi_{J_k}(\varrho)={\rm tr}_{I\setminus J_k}[\Phi_I(\varrho)]$ for all input states $\varrho$.

We note that channel compatibility is a natural generalization of measurement compatibility, i.e., the property of having a simultaneous readout for many measurements \cite{HeMi2017}. This follows directly from identifying measurements as channels with classical outputs and noting that the broadcast channel corresponds to the simultaneous readout. On top of being more general, channel compatibility differs from measurement compatibility in a crucial manner: Channels can be incompatible with themselves. A channel $\Phi_{A\to B}$ is {\it $n$-extendible}, i.e., it is $n$ times compatible with itself, if it can be broadcasted $n$ times by some channel. These channels are very well known in the literature (see, e.g. Refs.~\cite{Pankowski2013,HeMi2017,Berta2018,Kaur2019} and Refs. therein).
A channel is 2-extendible if and only if it is a post-processing of its conjugate channel~\cite{Pankowski2013,HeMi2017}. Such channels are also called {\it antidegradable}. Channels that are $n$-extendible for any $n\geq2$ are exactly the entanglement breaking channels (or measure-and-prepare channels) \cite{HeMi2017}.

Similarly to the marginal problem, channel incompatibility has a natural quantifier called the generalized robustness. For a channel tuple $\bm{\Lambda}$ one defines the generalized robustness $\mathcal{R}_F^g(\bm{\Lambda})$ with respect to compatible tuples $F$ as
\begin{equation}\label{Eq:ChanRob}
\mathcal{R}_F^g(\bm{\Lambda}) = \min\qty{t\geq 0\bigg\vert \frac{\bm{\Lambda} + t\bm{\Gamma}}{1+t}\in F},
\end{equation}
where the optimization is over channel tuples $\bm{\Gamma}$. For a single channel, one can take the set $F$ to be channels that are $n$-extendible.

\section{The channel-state dualism}\label{sec:ChStDu} In this paper, we mainly work with finite-dimensional Hilbert spaces. We present the methodology for this setting in the main text and discuss the straight-forward generalization to the infinite-dimensional case in Appendix~G. The infinite-dimensional case is only needed for the Gaussian scenario.

Whenever $\varrho_A=\sum_n t_n\dyad{n}$ is a full-rank state on $\hil_A$ with the {\it canonical purification} $\ket{\Omega_A}:=\sum_n\sqrt{t_n}\ket{nn}$, we denote for all channels $\Phi_{A\to B}$ the {\it $\ket{\Omega_A}$-Choi state} as $S_{\ket{\Omega_A}}(\Phi_{A\to B}):=({\rm id}\otimes\Phi_{A\to B})(\dyad{\Omega_A})$. Recall that, traditionally in the finite-dimensional setting, the Choi-state of a channel $\Phi_{A\to B}$ is defined w.r.t.\ the maximally entangled vector $\ket{\Omega_A}=d_A^{-1/2}\sum_{n=0}^{d_A-1}\ket{nn}$. In the infinite-dimensional case however, this vector is not available and we have to content ourselves with the Choi-states defined w.r.t.\ more general maximally entangled vectors. According to Refs.~\cite{KiBuUoPe2017,Haapasalo2019b}, the map $S_{\ket{\Omega_A}}$ is a well defined bijection between the set of channels $\Phi_{A\to B}$ and the set of bipartite states $\varrho_{AB}$ with the fixed first marginal ${\rm tr}_B[\varrho_{AB}]=\varrho_A$.

Conversely, when $\varrho_{AB}$ is a state on $\hil_A\otimes\hil_B$, we can assume that the reduced state $\varrho_A={\rm tr}_B[\varrho_{AB}]$ is of full rank by restricting the dimension of the subsystem $A$. By giving $\varrho_A$ the canonical purification $\ket{\Omega_A}$, we call the unique channel $\Phi_{A\to B}(\varrho):={\rm tr}_A[\varrho_{AB}(\varrho_A^{-1/2}\varrho^{T_A}\varrho_A^{-1/2}\otimes\id_{\hil_B})]$ the {\it $\ket{\Omega_A}$-Choi channel of $\varrho_{AB}$}. Here $\varrho^{T_A}$ is the partial transpose of $\varrho$ on the system $A$.

\section{Main result} Having an equal marginal on Alice's side is a necessary condition for the existence of a solution to the marginal problem. As this amounts to fixing the mapping in the Choi-Jamio\l kowski isomorphism, we are ready to state our main result.

\begin{theorem}\label{theor:central}
Fix the Choi-Jamio\l kowski isomorphism, i.e.,\ the full-rank state on $\hil_A$ and its canonical purification. A set of channels is compatible if and only if the marginal problem involving their Choi states has a solution. Moreover, this connection is quantitative in that the consistent robustness of the marginal problem in Eq.~\eqref{Eq:MarginalRobustness} matches with the incompatibility robustness of the channels in Eq.~\eqref{Eq:ChanRob}.
\end{theorem}

The proof can be found in Appendix~A. Note that as a special case of the result one gets a connection between symmetric extendibility and self-compatibility of channels.

\section{From states to channels: quantum memories}
A prominent way of quantifying the strength of quantum memories is by asking how close the corresponding quantum channel is to a measure-and-prepare channel. The generalised robustness with respect to measure-and-prepare channels gives a reasonable measure of such distance, as it has the basic properties one expects from a resource quantifier, i.e., faithfulness, monotonicity, convexity, and stability under tensor products~\cite{1907.02521}. The quantifier has also an operational meaning as the amount of advantage a quantum memory can give over memoryless channels in correlation tasks~\cite{1907.02521,uola2019b}. 
The quantifier is simply called the robustness of quantum memories (RoQM). Although having many desired properties, RoQM has one drawback: general methods for its efficient evaluation remain unknown. We note that approximate methods were developed in Ref.~\cite{1907.02521}.

Here, we provide a strategy for the efficient evaluation of the RoQM. Measure-and-prepare channels are closely related to separable states through the Choi-Jamio\l kowski isomorphism. As the set of separable states can be characterised with a converging hierarchy of SDPs, Theorem \ref{theor:central} can be used to develop a hierarchy of SDPs converging to the RoQM. On the round $n$ of the hierarchy, one calculates the robustness with respect to $n$-extendible channels. As $n$-extendible channels form a superset of $n+1$-extendible channels, every round of the hierarchy gives a lower bound on the next one, see Fig.~\ref{fig:depol1}. The hierarchy converges to the RoQM, as infinitely many times extendible channels coincide with measure-and-prepare channels~\cite{HeMi2017}, see Appendix~C for details.

\begin{observation}\label{prop:hierarchy}
The robustness of quantum memories can be evaluated with a converging hierarchy of SDPs.
\end{observation}

\begin{figure}[t!]
    \centering
    \includegraphics[width=.46\textwidth]{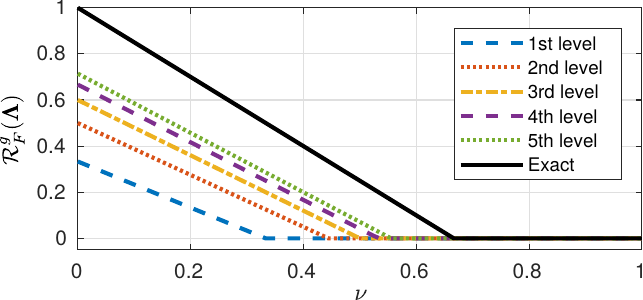}
    \caption{Generalized robustness of depolarizing channel $\mathbf{\Lambda}(\rho) = (1-\nu)\rho+\nu\frac{1}{d}\id_d$ for $d=2$ with respect to the set $F$ of measure-and-prepare channels. Solutions for $1-5$ levels of hierarchy~\cite{Doherty2002,Doherty2004} are compared to the exact solution obtained with PPT condition~\cite{PPT}.The hierarchy of ~\cite{Doherty2002,Doherty2004} becomes computationally demanding for higher levels. In practice, a PPT condition can be added to each level of this hierarchy, which in our case would lead to its convergence at level $1$.}
    \label{fig:depol1}
\end{figure}

We note that every level of the hierarchy can be evaluated from correlations between inputs to the channel and the measurement outcomes on the output~\cite{1907.02521,uola2019b}. These correlations correspond to decoding (output measurement) an encoded message (input state) after it has gone through the channel. For completeness, we illustrate this technique in Appendix~B for both $n$-self-compatibility and symmetric extendibility. This results in an experimentally feasible scenario for investigating the strength of quantum memories. Indeed, similar techniques have been used for experimental evaluation of the robustness of quantum steering and coherence~\cite{Sun_2018,Zheng2018}.

\section{From states to channels: entropic criteria} For states of low dimension the quantum marginal problem, and hence compatibility of channels (see Theorem~\ref{theor:central}), can be tackled with SDPs~\cite{boyd2004convex, Gaertner2012}. However, for higher-dimensional cases this approach becomes computationally demanding. One can nevertheless give some approximate solutions by means of entropic constraints, see, e.g., Ref.~\cite{carlen2013extension}. Now we demonstrate the use of Theorem \ref{theor:central} by translating the basic entropic results into witnesses of channel incompatibility. 

\begin{observation}
Incompatibility of channels can be determined by inequalities satisfied by von Neumann entropy.
\end{observation}

Indeed, let us define the $\ket{\Omega_A}$-entropy of a channel $\Phi_{A\to B}$ as the von Neumann entropy of its $\ket{\Omega_A}$-Choi state
$H_{\ket{\Omega_A}}(\Phi_{A\to B}) = -\tr[S_{\ket{\Omega_A}}(\Phi_{A\to B})\log S_{\ket{\Omega_A}}(\Phi_{A\to B})].$
Entropy of the reduced states of $\ket{\Omega_A}$-Choi state of $\Phi_{A\to B}$ are simply von Neumann entropy of state $\varrho_A$, $H(\varrho_A)$ for which $\ket{\Omega_A}$ is a canonical purifying vector, and von Neumann entropy of $\varrho_B = \Phi_{A\to B}(\varrho^T_A)$, where transpose is taken with respect to the basis in which $\ket{\Omega_A}$-Choi state of $\Phi_{A\to B}$ is defined.

Von Neumann entropy is known to satisfy certain linear inequalities, namely, strong subadditivity~\cite{lieb1973proof} and weak monotonicity (Eq.~(3.2) of Ref.~\cite{Araki1970}). The latter can be directly applied to our problem. It takes the following form for two channels $\Phi_{A\to B_1}$ and $\Phi_{A\to B_2}$
\begin{align}
\label{eq:entr_1}
     H_{\ket{\Omega_A}}(\Phi_{A\to B_1}) + &H_{\ket{\Omega_A}}(\Phi_{A\to B_2}) \notag\\
    & - H(\varrho_{B_1}) - H(\varrho_{B_2}) \geq 0.
\end{align}

\begin{figure}[t!]
    \centering
    \includegraphics[width=.48\textwidth]{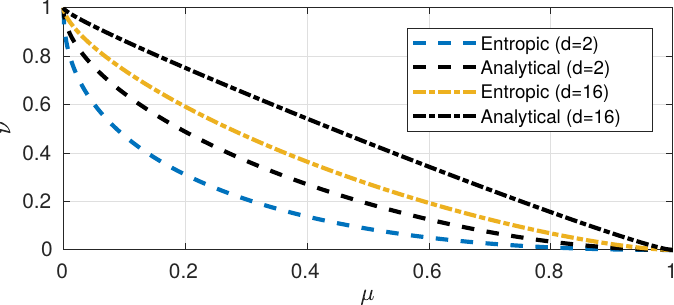}
    \caption{Areas of compatibility of two depolarizing channels with parameters $\mu$ and $\nu$ for two cases: $d=2$ and $d=16$.}
    \label{fig:depol}
\end{figure}

The above criteria are applicable to channels of arbitrary dimension and are experimentally testable e.g. by means of process tomography~\cite{altepeter2003ancilla}. In Fig.~\ref{fig:depol} we show the boundaries of the areas of compatibility of two depolarizing channels $\Phi_{A\to B_1}$ and $\Phi_{A\to B_2}$ defined as
\begin{align}
\Phi_{A\to B_1}(\varrho)&=(1-\mu)W(q,p)\varrho W(q,p)^\dagger+\frac{\mu}{d}\id,\label{eq:depolB1}\\
\Phi_{A\to B_2}(\varrho)&=(1-\nu)W(r,s)\varrho W(r,s)^\dagger+\frac{\nu}{d}\id,\label{eq:depolB2}
\end{align}
with $W(q,p)\ket{j}=e^{i\frac{\pi}{d}(q+2j)p}\ket{j+q}$, for $q,\,p,\,r,\,s\in\Z_d$ for dimensions $d=2$ and $d=16$ and compare those with analytical results of  Ref.~\cite{Haapasalo2019a} (also given in Eq.~\eqref{ineq:munu}). We note that the entropic criterion in Eq.~(\ref{eq:entr_1}) is not the only one that can be translated and that the criteria may depend on the choice of Choi isomorphism (see Appendix~D).

It is clear that the aforementioned entropic constraints can also be applied to the problem of symmetric extendibility (see Appendix~D). Moreover, the symmetric extendibility of a bipartite qubit state has been fully resolved in Ref.~\cite{Chen_et_al_2018} and this result readily characterizes self-compatibility and, hence, the antidegradability of any qubit-to-qubit channel. See the straight-forward translation in Appendix~E. A very similar analysis to characterize the antidegradable qubit-to-qubit channels has been carried out in~\cite{PaChen2017}.

\section{From channels to states: depolarizing noise and the Bell-diagonal marginal problem}
Consider two depolarizing channels $\Phi_{A\to B_1}$ and $\Phi_{A\to B_2}$ as defined in Eqs.~(\ref{eq:depolB1}) and (\ref{eq:depolB2}). The compatibility of these channels was completely characterized in Ref.~\cite{Haapasalo2019a}. Namely for $\mu,\,\nu\in[0,1]$, the channels  $\Phi_{A\to B_1}$ and $\Phi_{A\to B_2}$ are compatible if and only if
\begin{equation}\label{ineq:munu}
\mu+\frac{2}{d}\sqrt{\mu\nu}+\nu\geq1.
\end{equation}
Using Theorem~\ref{theor:central} we get the following result.

\begin{observation}\label{theor:Wstates}
Consider two pure states $\ket{\fii_{AB_1}}$ and $\ket{\fii_{AB_2}}$, such that the common marginal $\varrho_A$ is of full rank. For $\mu,\,\nu\in[0,1]$, there is a tripartite state $\varrho_{AB_1B_2}$ such that $\varrho_{AB_1}=(1-\mu)\dyad{\fii_{AB_1}}+\mu\frac{1}{d}\varrho_A\otimes\id_d$ and $\varrho_{AB_2}=(1-\nu)\dyad{\fii_{AB_2}}+\nu\frac{1}{d}\varrho_A\otimes\id_d$ if and only if inequality \eqref{ineq:munu} holds.
\end{observation}

The proof of the above Observation is presented in Appendix~F.

Another result concerns Pauli channels, which we denote by $\Phi_{\vec{p}}(\varrho) =  p_0\varrho+p_x\sigma_x\varrho\sigma_x^\dag+p_y\sigma_y\varrho\sigma_y^\dag+p_z\sigma_z\varrho\sigma_z^\dag$ for any probability vector $\vec{p} = (p_0,p_x,p_y,p_z)$. According to Ref.~\cite{Haapasalo2019a}, channels $\Phi_{\vec{p}}$ and $\Phi_{\vec{q}}$ are compatible if and only if there are real numbers $\alpha,\,\beta,\,\gamma$ such that ${\bf M}_{\vec{p},\vec{q}}(\alpha,\beta,\gamma)\geq0$, where
\begin{equation}\label{eq:Mpq}
{\bf M}_{\vec{p},\vec{q}}(\alpha,\beta,\gamma)=
\left(\begin{array}{cccc}
p_0&\alpha&\beta&\langle\vec{q}\rangle_1-\gamma\\
\cdot&p_x&\gamma&\langle\vec{q}\rangle_2-\beta\\
\cdot&\cdot&p_y&\langle\vec{q}\rangle_3-\alpha\\
\cdot&\cdot&\cdot&p_z
\end{array}\right),
\end{equation}
and $\langle\vec{q}\rangle_1=\frac{1}{2}(q_0-q_x-q_y+q_z)$, $\langle\vec{q}\rangle_2=\frac{1}{2}(q_0-q_x+q_y-q_z)$, and $\langle\vec{q}\rangle_3=\frac{1}{2}(q_0+q_x-q_y-q_z)$. Note that, since the diagonal entries of the above matrix are within the interval $[0,1]$, we may restrict the range of the parameters $\alpha$, $\beta$, and $\gamma$ within $[-1,1]$.

A two-qubit state $\varrho$ is called {\it Bell-diagonal} if there is a probability vector $\vec{p}$ such that $\varrho=\varrho_{\vec{p}}:=p_0\dyad{\Omega_0}+p_x\dyad{\Omega_x}+p_y\dyad{\Omega_y}+p_z\dyad{\Omega_z}$, where $\ket{\Omega_0}:=\frac{1}{\sqrt{2}}(\ket{00}+\ket{11})$ and $\ket{\Omega_r}:=(\id_2\otimes\sigma_r)\ket{\Omega_0}$ for $r=x,\,y,\,z$. The following Observation follows from Theorem~\ref{theor:central} and the above result on compatibility of Pauli channels and provides the solution to the marginal problem of two Bell-diagonal two-qubit states.

\begin{observation}\label{prop:Belldiag}
For probability vectors $\vec{p}$ and $\vec{q}$, there is a three-qubit state $\varrho_{AB_1B_2}$ such that $\varrho_{AB_1}=\varrho_{\vec{p}}$ and $\varrho_{AB_2}=\varrho_{\vec{q}}$ if and only if there are $\alpha,\,\beta,\,\gamma\in[-1,1]$ such that ${\bf M}_{\vec{p},\vec{q}}(\alpha,\beta,\gamma)\geq0$.
\end{observation}

Together with the symmetric extendibility result of Ref.~\cite{Chen_et_al_2018}, Observation \ref{prop:Belldiag} can be taken as a first step towards characterizing all those pairs of two-qubit states which are margins of a three-qubit state.

\section{Gaussian marginal problems and Gaussian broadcasting}

To extend our technique to continuous variable systems, we derive necessary and sufficient conditions for the solvability of marginal problems involving Gaussian states. Using this and Theorem \ref{theor:central} allows also the characterization of compatibility of Gaussian channels. This generalizes the compatibility results in Ref.~\cite{Lami_et_al}.

Recall that a Gaussian state $\varrho$ on an $N$-mode system is uniquely characterized by the {\it covariance matrix} $\vec{C}_\varrho$ associated with the mean squares of the quadrature operators and the {\it displacement vector} $\vec{d}_\varrho$ associated with the means of the quadratures. We concentrate here on the case of three systems $A$, $B_1$, and $B_2$ with modes $N_A$, $N_{B_1}$, and $N_{B_2}$; the general case is discussed in Appendix~G. We denote the special symplectic matrices of any $C\in\{A,B_1,B_2\}$ by
$$
\vec{S}_C:=\left(\begin{array}{cc}
     0&\id_{N_C} \\
     -\id_{N_C}&0
\end{array}\right).
$$
We are ready to present the solution to the Gaussian marginal problem, the proof of which is found in Appendix~G.

\begin{observation}\label{obs:GaussianMarg}
The marginal problem involving the Gaussian states $\varrho_{AB_1}$ and $\varrho_{AB_2}$ has a solution if and only if there is a Gaussian state $\varrho_{AB_1B_2}$ solving the problem. Moreover, if
$$
\vec{C}_{\varrho_{AB_i}}=\left(\begin{array}{cc}
     \vec{C}_A&\vec{X}_i  \\
     \vec{X}_i^T&\vec{C}_{B_i} 
\end{array}\right),\quad i=1,\,2,
$$
the marginal problem is solvable if and only if there is a real $(2N_{B_1}\times 2N_{B_2})$-matrix $\vec{Y}$ such that
$$
\left(\begin{array}{ccc}
    \vec{C}_A+i\vec{S}_A &\vec{X}_1 &\vec{X}_2  \\
    \vec{X}_1^T &\vec{C}_{B_1}+i\vec{S}_{B_1} &\vec{Y}  \\
    \vec{X}_2^T &\vec{Y}^T &\vec{C}_{B_2}+i\vec{S}_{B_2}
\end{array}\right)\geq0.
$$
\end{observation}

To translate the above result to channels, we recall that a channel $\Phi_{A\to B_i}$, $i=1,\,2$, is Gaussian if it maps Gaussian states into Gaussian states. Such channel corresponds to real matrices $\vec{K}$ $(2N_{B_i}\times 2N_{B_i})$ and $\vec{L}$ $(2N_A\times 2N_{B_i})$ and a vector $\vec{m}\in\mathbb{R}^{2N_{B_i}}$ (satisfying certain conditions stated in Appendix~G) such that $\Phi_{A\to B_i}=\Phi_{\vec{K},\vec{L},\vec{m}}$ maps a Gaussian state $\varrho$ into the Gaussian state $\sigma$ with $\vec{C}_\sigma=\vec{L}^T\vec{C}_\varrho\vec{L}+K$ and $\vec{d}_\sigma=\vec{L}^T\vec{d}_\varrho+\vec{m}$.

Using Theorem~\ref{theor:central} together with Observation~\ref{obs:GaussianMarg} we get the following Observation, the proof of which is in Appendix~G.

\begin{observation}\label{obs:GaussianComp}
Gaussian channels are compatible if and only if they have a Gaussian joint channel. Specifically, the channels $\Phi_{\vec{K}_i,\vec{L}_i,\vec{m}_i}$, $i=1,\,2$, are compatible if and only if there is a real $(2N_{B_1}\times 2N_{B_2})$-matrix $\vec{Z}$ such that
\begin{widetext}
$$
\left(\begin{array}{cc}
    \vec{K}_1+i\vec{S}_{B_1}-i\vec{L}_1^T\vec{S}_A\vec{L}_1 &\vec{Z}-i\vec{L}_1^T\vec{S}_A\vec{L}_2  \\
    \vec{Z}^T+i\vec{L}_2^T\vec{S}_A^T\vec{L}_1 &\vec{K}_2+i\vec{S}_{B_2}-i\vec{L}_2^T\vec{S}_A\vec{L}_2  
\end{array}
\right)\geq0.
$$
\end{widetext}
\end{observation}

\section{Conclusions}
We built a general connection between channel compatibility and marginal problems. This brings together seemingly different modes of incompatibility, i.e. incompatibility of states, measurements and channels. As a special case, this proves a connection between symmetric extendibility of states and self-compatibility of channels.

To demonstrate the usefulness of the connection, we translated various key results between the fields resulting in, e.g. entropic criteria for channel incompatibility, an efficient computational method for evaluating the strength of quantum memories, and solutions to various marginal problems.

The connection between compatibility and marginal problems as presented in Theorem \ref{theor:central} is also valid when the input system of the channels is a separable Hilbert space but not necessarily finite dimensional. The other Hilbert spaces may even be non-separable. Marginal problems are mainly studied in the finite-dimensional settings, but this connection allows one to approach infinite-dimensional marginal problems using the well-established methods of incompatibility in infinite-dimensional systems.

\section*{Acknowledgments}
RU is thankful for the financial support from the Finnish Cultural Foundation. NM acknowledges the financial support by First TEAM Grant No. 2016-1/5. TK acknowledges support by the DFG and the ERC (Consolidator Grant 683107/TempoQ). EH acknowledges support from the National Natural Science Foundation of China (Grant No. 11875110).

\section*{Note added}
While finishing this work, we became aware of a related approach to the channel marginal problem. Namely, the authors of Ref.~\cite{Hsieh2021} define the quantum channel marginal problem as the question whether various local dynamics can be seen as instances of the same global dynamics. The difference to our approach is that whereas we concentrate on channels having the same input system, Ref.~\cite{Hsieh2021} allows for various different inputs. As an example, in the approach of Ref.~\cite{Hsieh2021} two identity channels are compatible, when they don't have an overlapping input system, whereas identity channels are by default incompatible in our approach due to the impossibility of perfect broadcasting. To conclude, one could say that our approach is more directly motivated from the perspective of generalising measurement incompatibility, whereas the approach of Ref.~\cite{Hsieh2021} is more oriented towards generalizing the quantum marginal problem.

\bibliographystyle{unsrt}

\onecolumn\newpage
\appendix

\section{Formal statement and proof of Theorem~\ref{theor:central}}
\label{appA}
In this appendix we give a more formal statement of our central result, Theorem \ref{theor:central}, with a proof. We also give some small additional results alluded to in this paper this far.

Let us first derive the form of the Choi channel for a given bipartite state given in the end of Section \ref{sec:ChStDu}. Let us assume that $\Phi_{A\to B}$ is a channel and $\varrho_{AB}$ is a state on $\hil_{AB}$ such that this state and channel are connected with the Choi-Jamio\l kowski isomorphism associated with a canonical purification $\ket{\Omega_A}$ of the first margin $\varrho_A$ of $\varrho_{AB}$. We may assume that there are $t_n>0$ such that $\sum_n t_n=1$ and an orthonormal basis $\{\ket{n}\}_n$ of $\hil_A$ such that $\ket{\Omega_A}=\sum_n\sqrt{t_n}\ket{nn}$. Denote the transpose defined w.r.t. the basis $\{\ket{n}\}_n$ by $D\mapsto D^{T_A}$. Recall that we may define the Heisenberg dual $\Phi_{A\to B}^*$ (a completely positive (normal) linear map taking bounded operators of the output $\hil_B$ to those of the input $\hil_A$) of $\Phi_{A\to B}$ through $\tr[\varrho\Phi_{A\to B}^*(D)]=\tr[\Phi_{A\to B}(\varrho)D]$. We have, for any bounded operator $D$ on $\hil_A$ and $E$ on $\hil_B$,
\begin{align*}
\tr[\varrho_{AB}(D\otimes E)]&=\tr[({\rm id}\otimes\Phi_{A\to B})(\dyad{\Omega_A})(D\otimes E)]=\bra{\Omega_A}D\otimes\Phi_{A\to B}^*(E)\ket{\Omega_A}\\
&=\sum_{m,n}\sqrt{t_m t_n}\bra{m}D\ket{n}\bra{m}\Phi_{A\to B}^*(E)\ket{n}=\sum_{m,n}\sqrt{t_m t_n}\bra{n}D^{T_A}\ket{m}\bra{m}\Phi_{A\to B}^*(E)\ket{n}\\
&=\tr[\varrho_A^{1/2}D^{T_A}\varrho_A^{1/2}\Phi_{A\to B}^*(E)]=\tr[\Phi_{A\to B}(\varrho_A^{1/2}D^{T_A}\varrho_A^{1/2})E].
\end{align*}
Note that this calculation holds even in the case where $\hil_A$ (and $\hil_B$) is infinite dimensional and separable. If ${\rm dim}\,\hil_A<\infty$, we may substitute $D=\varrho_A^{-1/2}\varrho^{T_A}\varrho_A^{-1/2}$, and we obtain the form of the Choi channel as described in Section \ref{sec:ChStDu}. Note also that, even when $\hil_A$ is separable but not finite dimensional, $\Phi_{A\to B}$ is fully determined by evaluating it on inputs $\varrho_A^{1/2}D\varrho_A^{1/2}$ with bounded operators $D$ on $\hil_A$ whenever $\varrho_A$ is faithful (which is the proper counterpart of being of full rank in infinite dimensions) since this set is dense in the trace class of $\hil_A$ w.r.t. the trace norm. We may now formalize Theorem~\ref{theor:central}.\\

\noindent {\bf Theorem 1. (formal)} Denote $I:=\{A,B_1,\ldots,B_n\}$ for $n\in\N$ and $J_j:=\{A,B_j\}$, $j=1,\ldots,\,n$, where $A$ is associated with the Hilbert space $\hil_A$ and $B_j$ with $\hil_{B_j}$, $j=1,\ldots,\,n$. Moreover, we define $\hil_j:=\hil_A\otimes\hil_{B_j}$ and $\hil:=\hil_A\otimes\hil_{B_1}\otimes\cdots\otimes\hil_{B_n}$.
\begin{itemize}
    \item[(i)] Let $\varrho_A$ be a full-rank (or, in the infinite-dimensional case, faithful) state on $\hil_A$ and $\ket{\Omega_A}$ be a canonical purification for $\varrho_A$. Channels $\Phi_{A\to B_j}$ from $\hil_A$ to $\hil_{B_j}$, $j=1,\ldots,\,n$, are compatible if and only if there is a state $\varrho$ on $\hil$ such that ${\rm tr}_{I\setminus J_j}[\varrho]=S_{\ket{\Omega_A}}(\Phi_{A\to B_j})$ for $j=1,\ldots,\,n$.
    \item[(ii)] Let $\varrho_j$ be a state over $\hil_j$ for all $j=1,\ldots,\,n$. There is a state $\varrho$ on $\hil$ such that ${\rm tr}_{I\setminus J_j}[\varrho]=\varrho_j$ if and only if ${\rm tr}_{B_j}[\varrho_j]=\varrho_A$ is fixed for all $j=1,\ldots,\,n$ and, upon assuming that $\varrho_A$ is of full rank (or faithful) and picking a canonical purification $\ket{\Omega_A}$ for $\varrho_A$, the $\ket{\Omega_A}$-Choi channels of $\varrho_j$, $j=1,\ldots,\,n$, are compatible.
    \item[(iii)] Let $\Phi_{A\to B_j}$ be channels from $\hil_A$ to $\hil_{B_j}$ and $\varrho_j$ be states over $\hil_j$ sharing the common (full-rank or faithful) $A$-margin $\varrho_A$ for $j=1,\ldots,\,n$, and pick a canonical purification $\ket{\Omega_A}$ for $\varrho_A$. Whenever $\rho_j=S_{\ket{\Omega_A}}(\Phi_{A\to B_j})$ for $j=1,\ldots,\,n$, the incompatibility robustness of $(\Phi_{A\to B_1},\ldots,\Phi_{A\to B_n})$ coincides with the consistent marginal robustness of $(\varrho_1,\ldots,\varrho_n)$.
\end{itemize}

\begin{proof}
Let us prove item (i): Let us first assume that $\Phi_{A\to B_j}$ are compatible and fix a joint channel $\Phi$ for them; recall that the input space of $\Phi$ is $\hil_A$ and the output space is $\hil_{B_1}\otimes\cdots\otimes\hil_{B_n}$. Denote $\varrho_j:=S_{\ket{\Omega_A}}(\Phi_{A\to B_j})$, $j=1,\ldots,\,n$, and $\varrho:=S_{\ket{\Omega_A}}(\Phi)$. Denote the identity operator on $\hil_{B_j}$ by $\id_j$. Using the above derivation for the Choi channel, we have, for all $j=1,\ldots,\,n$,
\begin{eqnarray*}
\tr[\varrho_j(D\otimes E)]&=&\tr[\Phi_{A\to B_j}(\varrho_A^{1/2}D^{T_A}\varrho_A^{1/2})E]={\rm tr}\big[\Phi(\varrho_A^{1/2}D^{T_A}\varrho_A^{1/2})(\id_1\otimes\cdots\otimes\id_{j-1}\otimes E\otimes\id_{j+1}\otimes\cdots\otimes\id_n)\big]\\
&=&\tr[\varrho(D\otimes\id_1\otimes\cdots\otimes\id_{j-1}\otimes E\otimes\id_{j+1}\otimes\cdots\otimes\id_n)]=\tr[{\rm tr}_{I\setminus J_j}[\varrho](D\otimes E)].
\end{eqnarray*}
for all bounded operators $D$ on $\hil_A$ and $E$ on $\hil_{B_j}$. Thus, $\varrho_j={\rm tr}_{I\setminus J_j}[\varrho]$, $j=1,\ldots,\,n$. The proof of the converse statement is contained in the proof of item (ii).

Let us go on to proving item (ii): Note that, for the existence of a joint state $\varrho$ of the claim, it is necessary that the $A$-margins of the states $\varrho_j$, $j=1,\ldots,\,n$, coincide. According to our earlier observation, we may freely assume that this shared margin $\varrho_A$ is of full rank (or faithful) and we may fix a canonical purification $\ket{\Omega_A}$ for it. Assume first that there is $\varrho$ such that ${\rm tr}_{I\setminus J_j}[\varrho]=\varrho_j$. Denote, for each $j=1,\ldots,\,n$, by $\Phi_{A\to B_j}$ the channel such that $S_{\ket{\Omega_A}}(\Phi_{A\to B_j})=\varrho_j$ and by $\Phi$ the channel (with the input $\hil_A$ and output $\hil_{B_1}\otimes\cdots\otimes\hil_{B_n}$) such that $S_{\ket{\Omega_A}}(\Phi)=\varrho$. Denote, again, the identity operator on $\hil_{B_i}$ by $\id_i$ and pick $j\in\{1,\ldots,\,n\}$. For any bounded $D$ on $\hil_A$ and $E$ on $\hil_{B_j}$,
\begin{eqnarray*}
&&\tr[{\rm tr}_{I\setminus J_j}[\Phi(\varrho_A^{1/2}D\varrho_A^{1/2})]E]={\rm tr}\big[\Phi(\varrho_A^{1/2}D\varrho_A^{1/2})(\id_1\otimes\cdots\otimes\id_{j-1}\otimes E\otimes\id_{j+1}\otimes\cdots\otimes\id_n)\big]\\
&&=\tr[\varrho(D^{T_A}\otimes\id_1\otimes\cdots\otimes\id_{j-1}\otimes E\otimes\id_{j+1}\otimes\cdots\otimes\id_n)]=\tr[\varrho_j(D^{T_A}\otimes E)]=\tr[\Phi_{A\to B_j}(\varrho_A^{1/2}D\varrho_A^{1/2})E].
\end{eqnarray*}
According to our discussion after the derivation of the Choi channel in the beginning of this appendix, this implies $\Phi_{A\to B_j}(\varrho)={\rm tr}_{I\setminus J_j}[\Phi(\varrho)]$ for all states $\varrho$ on $\hil_A$ and $j=1,\ldots,\,n$. The proof of the converse statement follows from the proof of item (i).

The item (iii) follows from the observation that the Choi-Jamio\l kowski isomorphism $S_{\ket{\Omega_A}}$ is an affine bijection between the set of channels $\Phi_{A\to B}$ and the set of states $\varrho_{AB}$ such that ${\rm tr}_B[\varrho_{AB}]=\varrho_A$ is fixed (and of full rank or faithful). Thus, all the convex structures of these sets are mapped in a one-to-one fashion and, particularly, the two robustness measures coincide.
\end{proof}

\section{Operational interpretation}\label{AppD}

The generalised robustness corresponds to a convex distance from a free or resourceless set $F$. Importantly, it also quantifies the outperformance a resource provides over the free set in correlation based games \cite{Ryuji2019,TaRe2019,uola2019a,uola2019b}. Note that in what follows in this section, we concentrate on self-compatibility of channels and symmetric extendibility of states. This is simply because the techniques are more compactly presentable in this scenario as opposed to channel compatibility and the marginal problem. We note that the techniques do generalise to these cases by following, e.g., \cite{uola2019b}.

An input-output game $G$ is a tuple consisting of an input ensemble $\{\sigma_a\}$, a POVM $\{N_b\}$ on the output, and a real-valued reward function $\{\omega_{ab}\}$ assigning a score to each input-output pair $(a,b)$. The players aim at communicating the input states using a quantum channel $\Lambda$ that maximizes the payoff
\begin{align*}
    P(\Lambda,G):=\sum_{a,b}\omega_{ab}\text{tr}[\Lambda(\sigma_a)N_b].
\end{align*}
We note that the payoff is covariant under scaling and shifting of the reward function $\{\omega_{ab}\}$. To remove this covariance (i.e. to treat all games on an equal footing), we define a canonical form of a game as one having the minimum payoff equal to zero and the maximal payoff equal to one, when optimized over quantum channels. The reason for using canonical games is that they provide a tight connection between optimal payoff and robustness measure, as is shown below.

Using an optimal solution of the robustness optimization problem in Eq.~(\ref{Eq:ChanRob}) for single channels (i.e. the free set is the set of $n$-extendible channels) one can write $\Lambda=(1+\mathcal{R}_F^g[\Lambda])\Lambda_F-\mathcal{R}_F^g[\Lambda]\Gamma$, where $\Lambda_F\in F$ and $\Gamma$ is some channel. Noting that in the canonical form the input-output games have non-negative payoffs and that the payoff is linear in the first argument, we get $P(\Lambda,G)\leq (1+\mathcal{R}_F^g[\Lambda])\max_{\Xi\in F} P(\Xi,G)$. Crucially, this upper bound can be reached with an appropriate choice of the game. To see this, we first note that the robustness can be written in the Choi picture as
\begin{equation}\label{Eq:ChoiChanRob}
\mathcal{R}_F^g(\Lambda) = \min\qty{t\geq 0\bigg\vert \frac{S_\Lambda + tS_\Gamma}{1+t}=\Lambda_F\in S_F},
\end{equation}
where $S$ refers to the canonical Choi map and the set $S_F$ is the image of the free set under this map. By solving $S_\Gamma=\frac{1}{t}(S_{\tilde\Lambda}-S_\Lambda)$, where $\tilde\Lambda=(1+t)\Lambda_F$, one gets
\begin{align*}
    1+\mathcal{R}_F^g[\Lambda] = \min_{S_{\tilde\Lambda}\in C_{S_F}}\, & \tr[S_{\tilde\Lambda}] \\
\text{s.t.: }& S_{\tilde\Lambda}\geq S_\Lambda,\notag
\end{align*}
where $S_\Phi$ is the canonical Choi picture, i.e. the one using the isotropic state, of a channel $\Phi$, and $C_K$ is the conic hull of a set $K$, that is the set one gets by multiplying the set $K$ with the non-negative part of the real line. This is an instance of a cone program and as such it has the dual form~\cite{uola2019a, Gaertner2012}
\begin{align}\label{ChanRobDuality}
    1+\mathcal{R}_F^g[\Lambda] = \max_{W\geq0}\, & \tr[S_\Lambda W] \\
\text{s.t.: }& \tr[S_\Psi W]\leq1\ \forall \ S_\Psi\in S_F.\notag
\end{align}
For scenarios with the Slater condition being satisfied the optimal value of the primal and the dual problem coincide. The validity of the Slater condition corresponds to finding a $S_\Gamma\in C_{S_F}$ such that $S_\Gamma> S_\Lambda$. This is clearly true whenever there is a full-rank point in the cone $C_{S_F}$. As the completely mixed state is in the cone we consider, we conclude that the optimal value of the primal and the dual problem coincide.

Crucially, any positive operator $W$ on the shared system can be written as $W=\sum_{a,b}\omega_{ab}\sigma_a^T\otimes N_b$ with $\omega_{ab}$ being real numbers, $(\cdot)^T$ denoting the transposition in the computational basis, and the sets $\{\sigma_a\}$ and $\{N_b\}$ forming an ensemble and a POVM respectively. Using the inverse of the Choi-Jamio\l kowski map, this shows that the objective function of Eq.~(\ref{ChanRobDuality}) is an instance of an input-output game for each choice of $W$. As from Eq.~(\ref{Eq:ChanRob}) it is clear that for an optimal witness $W$ one has $\tr[W S_\Gamma]=0$, an optimal solution of Eq.~(\ref{ChanRobDuality}) forms an instance of a canonical correlation game up to scaling. We arrive at the following formula noting that the left-hand-side is independent of scaling
\begin{align*}
    \sup_{G}\frac{P(\Lambda,G)}{\max_{\Xi\in F}P(\Xi,G)}=1+\mathcal{R}_F^g[\Lambda].
\end{align*}

It should be mentioned that on each level $n$ of self-compatibility, the conic program above reduces to a semi-definite program. These programs are efficiently solvable and they result in an optimal witness $W$. Using the above calculation, one can easily find an implementation for these optimal witnesses as an input-output game. In this way, although evaluating the robustness of quantum memories remains a hard task, on each level of our converging hierarchy, one gets better lower bounds on the robustness that have a simple implementation as an input-output game.

To see how the above can be applied to symmetric extendibility, we call correlation game $G_{cor}$ a tuple consisting of a POVM $\{M_a\}$ on Alice, a POVM $\{N_b\}$ on Bob, and a real-valued reward function $\omega_{ab}$ assigning a score to each pair of outcomes $(a,b)$. The task of the players is to prepare a shared state $\varrho_{AB}$ that maximizes the following payoff
\begin{align*}
    P(\varrho_{AB},G_{cor}):=\sum_{a,b}\omega_{ab}\text{tr}[(M_a\otimes N_b)\varrho_{AB}].
\end{align*}
As in the case of input-output games, the payoff is covariant under scaling and shifting of the reward function. This motivates the definition of a canonical correlation game as one that has minimum payoff equal to zero and the maximal payoff equal to one when optimized over shared states.

As symmetric extendibility can be seen as a property of a single state, we shorten the notation by considering Eq.~(\ref{Eq:MarginalRobustness}) for individual shared states instead of sets of states. To simplify the discussion further, we take a relaxed version of the consistent robustness, i.e. the generalized robustness $\mathcal{R}_F^g$ defined analogously to the case of channels in Eq.~(\ref{Eq:ChanRob}). In case of states this becomes Eq.~(\ref{Eq:MarginalRobustness}) with the modification that the states $\bm{\tau}$ do not have to share a common marginal with $\varrho_{AB}$.

As for input-output games, that the payoff is linear in the first argument and one has $P(\varrho_{AB},G_{cor})\leq (1+\mathcal{R}_F^g[\varrho_{AB}])\max_{\sigma_{AB}\in F} P(\sigma_{AB},G_{cor})$. 
The upper bound can again be reached with a proper choice of the game. To see this, we write the robustness as
\begin{align*}
    1+\mathcal{R}_F^g[\varrho_{AB}] = \min_{\tilde\sigma_{AB}\in C_F}\, & \tr[\tilde\sigma_{AB}] \\
\text{s.t.: }& \tilde\sigma_{AB}\geq\varrho_{AB},\notag
\end{align*}
where $C_F$ is the conic hull of $F$. This is again an instance of a cone program and its dual reads
\begin{align*}
    1+\mathcal{R}_F^g[\varrho_{AB}] = \max_{W\geq0}\, & \tr[\varrho_{AB}W] \\
\text{s.t.: }& \tr[\sigma_{AB}W]\leq1\ \forall \ \sigma_{AB}\in F.\notag
\end{align*}
The Slater condition is valid as the completely mixed state is a free state which can be scaled to be strictly larger than $\varrho_{AB}$.

In the case of states, we write the positive operator $W$ on the shared system as $W=\sum_{a,b}\omega_{ab}M_a\otimes N_b$ with $\omega_{ab}$ being real numbers and the sets $\{M_a\}$ and $\{N_b\}$ forming POVMs. As from Eq.~(\ref{Eq:MarginalRobustness}) it is clear that for an optimal witness $W$ one has $\tr[W\tau_{AB}]=0$, the optimal solution forms an instance of a canonical correlation game up to scaling. We arrive at the following formula noting that the left-hand-side is independent of scaling
\begin{align*}
    \sup_{G_{cor}}\frac{P(\varrho_{AB},G_{cor})}{\max_{\sigma_{AB}\in F}P(\sigma_{AB},G_{cor})}=1+\mathcal{R}_F^g[\varrho_{AB}]
\end{align*}

Symmetric extendibility can be characterised through hierarchy of semi-definite programs in which the level $n$ of the hierarchy checks for the existence of $n$ symmetric extensions~\cite{Doherty2002, Doherty2004}. Choosing the set $F$ to be those shared states that have $n$ or less symmetric extensions gives a task-oriented characterization for the level $n$ of the hierarchy. The statement holds true also for the limit of infinite extensions, i.e. for separability. Hence, whereas it is not clear whether states with $n$ symmetric extensions can outperform states with $n+1$ extensions in correlation experiments based on quantum steering and non-locality, we have derived a task-oriented characterization for all levels of the hierarchy in terms of correlation games.

As mentioned above, the correlation game characterization can be given for the marginal problem as well. To see this, one can perform all the above calculations by treating pairs (or sets) of shared states and payoff functions as direct sums, cf. \cite{uola2019b}.

\section{Proof of Observation \ref{prop:hierarchy}}

According to our main result, Theorem \ref{theor:central}, the set ${\bf EB}$ of entanglement-breaking channels and the set of separable states $\mathcal{S}_{\rm sep}$ are in one-to-one correspondence set up by the channel-state duality (with a fixed generating full-rank state). Similarly, for each $n=1,\,2,\ldots$, the set ${\bf Ch}_n$ of channels $n$ times compatible with themselves and the set $\mathcal{S}_n$ of bipartite states with $n$ symmetric extensions are in one-to-one correspondence. We concentrate on the finite-dimensional case where all these sets are compact (with respect to the natural topologies; the trace norm topology and the CB-norm topology), $\mathcal{S}_{n}\supset\mathcal{S}_{n+1}$ and ${\bf Ch}_n\supset{\bf Ch}_{n+1}$ for all $n=1,\,2,\ldots$, and
$$
\bigcap_{n=1}^\infty \mathcal{S}_n=\mathcal{S}_{\rm sep},\quad\bigcap_{n=1}^\infty {\bf Ch}_n={\bf EB}.
$$

Let $\Phi$ be a quantum memory (a channel) and denote its general noise robustness w.r.t.\ ${\bf Ch}_n$ by $R_n$. Recall that this means that $R_n$ is the minimum of numbers $t\geq0$ such that there is a channel $\Psi$ such that $(1+t)^{-1}(\Phi+t\Psi)\in{\bf Ch}_n$. Since ${\bf EB}\subset{\bf Ch}_n$ for all $n$, we have $R_n\leq R$ where $R$ is the generalized robustness w.r.t.\ ${\bf EB}$, i.e.,\ the generalized robustness of the quantum memory resource of $\Phi$. This robustness is the minimum of numbers $t\geq0$ such that there is a channel $\Psi$ such that $(1+t)^{-1}(\Phi+t\Psi)\in{\bf EB}$. Moreover, since ${\bf Ch}_{n+1}\subset{\bf Ch}_n$, $R_n\leq R_{n+1}$ for all $n$. Denote by $R_0$ the supremum of $\{R_n\}_{n=1}^\infty$. Naturally, $R_0\leq R$. Pick, for each $n=1,\,2,\ldots$, a channel $\Psi_n$ such that
$$
\Phi_n:=\frac{1}{1+R_n}(\Phi+R_n\Psi_n)\in{\bf Ch}_n.
$$
Since the set of all channels (with fixed input and output systems) is compact in this finite-dimensional setting, by passing onto a subsequence, if necessary, we may assume that $\Psi_n\overset{n\to\infty}{\to}\Psi$ for some channel $\Psi$. As $R_n\overset{n\to\infty}{\to}R_0$, it easily follows that
$$
\Phi_n\overset{n\to\infty}{\to}\Phi_0:=\frac{1}{R_0+1}(\Phi+R_0\Psi)
$$
and, since $\Phi_n\in{\bf Ch}_n$ for all $n=1,\,2,\ldots$ and ${\bf EB}$ is closed, we have $\Phi_0\in{\bf EB}$. By the definition of the general noise robustness, we now have $R\leq R_0$. Thus, $R=R_0$, implying that $(R_n)_{n=1}^\infty$ converges to the general robustness of the quantum memory resource of $\Phi$.

\section{More entropic constraints}
\label{appC}
Further constraints can be derived by combining SSA and WM constraints in a way that the resulting inequalities do not contain the unknown parameters like $\ket{\Omega_A}$-entropy of global channels. This can be done with the techniques of linear programming (see e.g. Ref.~\cite{fritz2013entropic}).  In particular, the so-called Fourier-Motzkin algorithm can be used to eliminate variables from the given system of linear inequalities. Geometrically, this elimination corresponds to a projection in the space of the vector of variables. In our case, since e.g., the $\ket{\Omega_A}$-entropy of the global channel $\Phi_{A\to B_1,B_2,B_3}$ cannot be accessed, the corresponding variable needs to be eliminated from the system of inequalities given by SSA and WM constraints. Here we introduce a new constraint for compatibility of three channels $\Phi_{A\to B_1}$, $\Phi_{A\to B_2}$, and $\Phi_{A\to B_3}$  obtained from Fourier-Motzkin elimination:
\begin{eqnarray}
\label{eq:entr_2}
H_{\ket{\Omega_A}}(\Phi_{A\to B_1}) + H_{\ket{\Omega_A}}(\Phi_{A\to B_2}) + 2H_{\ket{\Omega_A}}(\Phi_{A\to B_3})-H(\varrho_{B_1})-H(\varrho_{B_2})-2H(\varrho_A)\geq 0.
\end{eqnarray}

When one tries to resolve the compatibility of channels via explicit methods of marginal problem, e.g. SDP, the choice of $\ket{\Omega_A}$ in Choi isomorphism does not play any role. This is, however, not the case when one applies the entropic constraints. For the case of two depolarizing channels the optimal Choi isomorphism corresponds to the maximally entangled state. For channels with other types of noise the optimal choice of $\ket{\Omega_A}$ can be different. For example, for the noise $\frac{1}{d-1}\id_{d-1} = \frac{1}{d-1}\sum_{i=1}^{d-1}\dyad{i}$ the optimal $\ket{\Omega_A} = \sum_{i=1}^{d}\sqrt{t_i}\ket{i}\otimes\ket{i}$ is the one with $t_d$ epsilon-small and the rest of $t_i$ close to $\frac{1}{d-1}$.

Naturally, the entropic constraints can be applied to the problem of self-compatibility of quantum channels (as a translation of the results on symmetric extendibility of states). Taking both marginal entropies and $\ket{\Omega_A}$-entropies of channels equal, we obtain conditions for self-compatibility of channel $\Phi_{A\to B}$: $H_{\ket{\Omega_A}}(\Phi_{A\to B})- H(\varrho_B)\geq 0$ and $2H_{\ket{\Omega_A}}(\Phi_{A\to B}) -H(\varrho_B)-H(\varrho_A)\geq 0$ from Eq.~\eqref{eq:entr_1} and Eq.~\eqref{eq:entr_2} respectively. Whenever $H(\varrho_A)\geq H(\varrho_B)$, the second constraint is tighter.

\section{Self-compatibility of channels}
\label{appB}
We now go on to study self-compatibility of channels. In the discussion earlier on the entropic bounds for the compatibility of channels, we have defined the $\ket{\Omega_A}$-entropy of a channel as von Neumann entropy of its $\ket{\Omega_A}$-Choi state. However, one can equivalently define it as the Shannon entropy of the $\ket{\Omega_A}$-Choi state spectrum. In that sense, the results on the entropic bounds provide nonlinear inequalities for spectra of 2-extendible channels. One would expect, however, that the explicit forms of these spectra can be very cumbersome to write. On the other hand, some spectral constrains for symmetric extendibility of two-qubit states are known~\cite{Chen_et_al_2018}, which we translate below to the problem of self-compatibility of channels.

In order to identify the spectrum of the Choi-state of a channel, recall that operators $K_i:\hil_A\to\hil_B$, $i=1,\,2,\ldots$, contribute a {\it Kraus decomposition} for an $A\to B$-channel $\Phi$ if $\Phi(\varrho)=\sum_iK_i\varrho K_i^\dag$ (where the series converges w.r.t.\ the trace norm). For any state $\varrho_A$ on $\hil_A$, $\Phi$ has a Kraus operators $K_i$ such that $\tr[\varrho_A K_i^\dag K_j]=0$ whenever $i\neq j$; see Section 3.1 of \cite{Haapasalo2019b}. In this case, we say that $K_i$ are $\varrho_A$-orthogonal. Whenever $\varrho_A$ is of full rank and $\ket{\Omega_A}$ is a canonical purification for $\varrho_A$, we have the spectral decomposition $S_{\ket{\Omega_A}}(\Phi)=\sum_i\dyad{w_i}$ where $\ket{w_i}=(\id_\hil\otimes K_i)\ket{\Omega_A}$ for any $\varrho_A$-orthogonal set $\{K_i\}_i$ of Kraus operators for $\Phi$ \cite[Proposition 1]{Haapasalo2019b}. Thus, the spectrum of $S_{\ket{\Omega_A}}(\Phi)$ consists of the numbers 
$$
\lambda_{\varrho_A}^\Phi(i):=\tr[K_i\varrho_A K_i^\dag]
$$
and the vector $\vec{\lambda}_{\varrho_A}^\Phi:=\big(\lambda_{\varrho_A}^\Phi(i)\big)_i$ is essentially independent of the particular $\varrho_A$-orthogonal set of Kraus operators for $\Phi$.

Let $A$, $B$, and $C$ be now qubit systems. According to Ref. \cite{Chen_et_al_2018}, a state $\varrho$ on $\hil_A\otimes\hil_B$ is symmetrically extendible, i.e.,\ there is a three-qubit state $\hil_{ABC}$ such that $\varrho_{AB}=\varrho=\varrho_{AC}$ if and only if $\tr[{\rm tr}_A[\varrho]^2]\geq\tr[\varrho^2]-4\sqrt{\det(\varrho)}$. Let $\Phi_{A\to B}$ be the Choi channel of $\varrho$, i.e.,\ $\varrho=S_{\ket{\Omega_A}}(\Phi_{A\to B})$ for a standard purification $\ket{\Omega_A}$ of ${\rm tr}_B[\varrho]$. The right-hand side of the above inequality can be written entirely in terms of the spectrum of $\varrho$ and, thus, in terms of the probability vector $\vec{\lambda}_{\varrho_A}^{\Phi_{A\to B}}$. Moreover, ${\rm tr}_A[\varrho]=\Phi_{A\to B}(\varrho_A)$. Thus, we have the following:

\begin{observation}\label{theor:self-comp}
A qubit-to-qubit channel $\Phi$ is 2-extendible if, for some (and, hence, for any) full-rank qubit state $\varrho_A$,
$$
\tr[\Phi(\varrho_A)^2]\geq\sum_{i}\lambda_{\varrho_A}^\Phi(i)^2-4\prod_{i}\sqrt{\lambda_{\varrho_A}^\Phi(i)}.
$$
\end{observation}

In particular, choosing $\varrho_A=\frac{1}{2}\id$ and a Hilbert-Schmidt-orthogonal set $\{K_i\}_{i}^R$ of Kraus operators for $\Phi$, $R\leq4$ being the Kraus rank of $\Phi$, the channel $\Phi$ is 2-extendible if and only if
$$
\tr[\Phi(\id)^2]\geq\sum_{i}\|K_i\|_{\rm HS}^4-\frac{16}{2^{R/2}}\prod_{i}\|K_i\|_{\rm HS}
$$
where, for any qubit operator $K$, the Hilbert-Schmidt norm is defined as $\|K\|_{\rm HS}:=\sqrt{\tr[K^\dag K]}$.

\section{Proof of Observation \ref{theor:Wstates}}
\label{appE}

\begin{proof}
Suppose that $\ket{\Omega_A}$ is a standard purifying vector for $\varrho_A$. It easily follows that there are unitaries $U_{B_1}$ and $U_{B_2}$ such that $\ket{\fii_{AB_1}}=(\id_d\otimes U_{B_1})\ket{\Omega_A}$ and $\ket{\fii_{AB_2}}=(\id_d\otimes U_{B_2})\ket{\Omega_A}$. Clearly the original states are margins of a tripartite state if and only if there is a tripartite state $\varrho_{AB_1B_2}$ such that $\varrho_{AB_1}=\varrho_\mu$ and $\varrho_{AB_2}=\varrho_\nu$ where $\varrho_\lambda:=(1-\lambda)\dyad{\Omega_A}+\lambda\frac{1}{d}\varrho_A\otimes\id_d$ for all $\lambda\in[0,1]$. Using the channel state dualism $S_{\ket{\Omega_A}}$, this problem is equivalent with finding those $\mu,\,\nu\in[0,1]$ such that the channels $\varrho\mapsto (1-\mu)\varrho+\mu\frac{1}{d}\id_d$ and $\varrho\mapsto (1-\nu)\varrho+\nu\frac{1}{d}\id_d$ are compatible. This happens, according to the above, if and only if the inequality \eqref{ineq:munu} holds.
\end{proof}

\section{Multipartite Gaussian marginal problems and compatibility questions}

We fix natural numbers $N_A$, $N_{B_1}$,\ldots, and $N_{B_n}$ and Hilbert spaces $\hil_A:=L^2(\R^{N_A})$ and $\hil_{B_i}:=L^2(\R^{N_{B_i}})$, $i=1,\ldots,\,n$. Whenever $C\subseteq\{A,B_1,\ldots,B_n\}$ and $\vec{q},\,\vec{p}\in\R^{N_C}$ we define the operator
$$
R_C(\vec{q},\vec{p}):=-\vec{q}^T\vec{P}_C+\vec{p}^T\vec{Q}_C
$$
where $\vec{Q}_C$ and $\vec{P}_C$ are the canonically conjugated position and momentum operators of the subsystem $C$ in vector form. In the sequel, whenever we have linear combinations of these operators, they are understood as the uniquely defined operators defined on the intersection of the domains of all these operators which is a dense subspace of $L^2(\R^{N_C})$. We also define
\begin{eqnarray*}
R(\vec{z})&:=&R_A(\vec{z}_A)\otimes\id_{B_1}\otimes\cdots\otimes\id_{B_n}+\id_A\otimes R_{B_1}(\vec{z}_{B_1})\otimes\id_{B_2}\otimes\cdots\otimes\id_{B_n}+\cdots\\
&&\cdots+\id_A\otimes\id_{B_1}\otimes\cdots\otimes\id_{B_{n-1}}\otimes R_{B_n}(\vec{z}_{B_n})
\end{eqnarray*}
for all $\vec{z}_A\in\R^{2N_A}$ and $\vec{z}_{B_i}\in\R^{2N_{B_i}}$, $i=1,\ldots,\,n$. We define the field operators $R_C(\vec{z})$ similarly for any subsystem $C$.

Any state on any of the subsystems or on the entire system is called {\it Gaussian} if its characteristic function is a Gaussian function. This means that a Gaussian state $\varrho$ over, say, system $C$ is such that the means $\langle R_C(\vec z)\rangle_\varrho$ and $\langle R_C(\vec{z})^2\rangle_\varrho$ are defined for all $\vec{z}\in\R^{2N_C}$ where $N_C$ is the sum of the number of modes of the individual subsystems in $C$, and the displacement vector is defined through
$$
\langle R_C(\vec{z})\rangle_\varrho=\vec{d}_\varrho^T\vec{S}_{C}\vec{z}, \quad\vec{z}\in\R^{2N_C},
$$
and the covariance matrix is defined through
$$
\langle\big(R_C(\vec{z})-\vec{d}_\varrho^T\vec{S}_{C}\vec{z} \id_C\big)^2\rangle_\varrho=\frac{1}{4}\vec{z}^T\vec{C}_\varrho
\vec{z}, \quad\vec{z}\in\R^{2N_C},
$$
where $\vec{S}_C=\bigoplus_{D\in C}\vec{S}_D$. The covariance matrix necessarily satisfies $\vec{C}_\varrho+i\vec{S}_C\geq0$.

We next present and prove a useful lemma. Reading the proof, one easily sees that the following lemma can be generalized to the extent that, for Hilbert spaces $\hil_1$ and $\hil_2$ and self-adjoint (unbounded) operators $A_i$ on $\hil_i$, $i=1,\,2$, and any state $\sigma$ on $\hil_1\otimes\hil_2$ such that the second moment of $A_i$ is defined in the marginal state $\sigma_{(i)}$, $i=1,\,2$, also the expectation value $\langle A\otimes B\rangle_\sigma$ is defined. However, we only need the following special case of this more general fact.

\begin{lemma}\label{lemma:ExpMarg}
Suppose that $\sigma$ is a state over $B_i\&B_j$ with $i<j$ such that $\langle R_{B_i}(\vec{z})^2\rangle_{\sigma_{B_i}}$ and $\langle R_{B_j}(\vec{w})^2\rangle_{\sigma_{B_j}}$ are defined for any $\vec{z}\in\R^{2N_{B_i}}$ and $\vec{w}\in\R^{2N_{B_j}}$. It follows that $\langle R_{B_i}(\vec{z})\otimes R_{B_j}(\vec{w})\rangle_{\sigma}$ is defined for all $\vec{z}\in\R^{2N_{B_i}}$ and $\vec{w}\in\R^{2N_{B_j}}$.
\end{lemma}

\begin{proof}
Pick $\vec{z}\in\R^{2N_{B_i}}$ and $\vec{w}\in\R^{2N_{B_j}}$. Denote by $\mathsf{P}_i$ the spectral measure of $R_{B_i}(\vec{z})$, by $\mathsf{P}_j$ the spectral measure of $R_{B_j}(\vec{w})$, by $\mathsf{P}_{i,j}$ the spectral measure of $R_{B_i}(\vec{z})\otimes R_{B_j}(\vec{w})$ and by $\mathsf{P}_i\otimes\mathsf{P}_j$ the PVM (projection-valued measure) on $\R^2$ such that $(\mathsf{P}_i\otimes\mathsf{P}_j)(X\times Y)=\mathsf{P}_i(X)\otimes\mathsf{P}_j(Y)$ for all Borel (or Lebesgue) subsets $X$ and $Y$ of $\R$. Using the fact that $\mathsf{P}_{i,j}=(\mathsf{P}_i\otimes\mathsf{P}_j)\circ f^{-1}$, where $f:\R^2\to\R$ is the multiplication map $f(x,y)=xy$, and the Cauchy-Schwarz inequality, we have
\begin{eqnarray*}
\int_\R |z|\,\tr[\sigma\mathsf{P}_{i,j}(dz)]&=&\int_{\R^2}|xy|\,\tr[\sigma(\mathsf{P}_i\otimes\mathsf{P}_j)\big(d(x,y)\big)]\\
&\leq&\sqrt{\int_{\R^2}x^2\,\tr[\sigma(\mathsf{P}_i\otimes\mathsf{P}_j)\big(d(x,y)\big)]\,\int_{\R^2}y^2\,\tr[\sigma(\mathsf{P}_i\otimes\mathsf{P}_j)\big(d(x,y)\big)]}\\
&=&\sqrt{\int_\R x^2\,\tr[\sigma\big(\mathsf{P}_i(dx)\otimes\id_{B_j}\big)]\,\int_\R y^2\,\tr[\sigma\big(\id_{B_i}\otimes\mathsf{P}_j(dy)\big)]}\\
&=&\sqrt{\int_\R x^2\tr[\sigma_{B_i}\mathsf{P}_i(dx)]\,\int_\R y^2\,\tr[\sigma_{B_j}\mathsf{P}_j(dy)]}\\
&=&\sqrt{\langle R_{B_i}(\vec{z})^2\rangle_{\sigma_{B_i}}\langle R_{B_j}(\vec{w})^2\rangle_{\sigma_{B_j}}}<\infty.
\end{eqnarray*}
This means that the expectation value $\langle R_{B_i}(\vec{z})\otimes R_{B_j}(\vec{w})\rangle_{\sigma}$ is defined.
\end{proof}

Using the above lemma and the quantum central limit theorem, we may prove the following theorem stating that the marginal problem involving Gaussian states sharing a common subsystem has a solution if and only if it has a solution in the form of a Gaussian state over the entire system. This is a generalization of Observation \ref{obs:GaussianMarg}.

\begin{theorem}\label{theor:GaussianMarg}
The marginal problem involving, for each $i=1,\ldots,\,n$, a Gaussian state $\varrho_i$ over the system $A\&B_i$ has a solution if and only if it has a Gaussian solution, i.e.,\ there is a Gaussian state $\varrho^G$ over $A\&B_1\&\cdots\&B_n$ such that $\varrho^G_{AB_i}=\varrho_i$ for all $i=1,\ldots,\,n$. Moreover, if, for each $i=1,\ldots,\,n$,
$$
\vec{C}_{\varrho_i}=\left(\begin{array}{cc}
    \vec{C}_A &\vec{X}_i  \\
    \vec{X}_i^T &\vec{C}_{B_i} 
\end{array}\right),
$$
the above marginal problem has a solution if and only if, for each $i,\,j\in\{1,\ldots,n\}$, $i<j$, there is a real $(2N_{B_i}\times 2N_{B_j})$-matrix $\vec{Y}_{i,j}$ such that
\begin{equation}\label{eq:GaussianMargCond}
\left(\begin{array}{ccccc}
    \vec{C}_A+i\vec{S}_A &\vec{X}_1 &\vec{X}_2 &\cdots &\vec{X}_n  \\
    \vec{X}_1^{T} &\vec{C}_{B_1}+i\vec{S}_{B_1} &\vec{Y}_{1,2} &\cdots &\vec{Y}_{1,n}  \\
    \vec{X}_2^T &\vec{Y}_{1,2}^T &\vec{C}_{B_2}+i\vec{S}_{B_2} &\cdots &\vec{Y}_{2,n}  \\
    \vdots &\vdots &\vdots &\ddots &\vdots  \\
    \vec{X}_n^T &\vec{Y}_{1,n}^T &\vec{Y}_{2,n}^T &\cdots &\vec{C}_{B_n}+i\vec{S}_{B_n}
\end{array}
\right)\geq0.
\end{equation}
\end{theorem}

\begin{proof}
Fix a state $\varrho$ over the entire system $A\&B_1\&\cdots\&B_n$ such that $\varrho_{AB_i}=\rho_i$ for all $i=1,\ldots,\,n$. We will prove the existence of the Gaussian state $\varrho^G$ of the claim using the quantum central limit theorem. In order to apply this theorem, we have to show that the first and, more importantly, the second moments of the quadrature operators are defined in the state $\rho$. We first look at the first moments and show that we may, essentially, disregard them. It immediately follows that $\langle R(\vec{z})\rangle_{\varrho}$ is defined for all $\vec{z}=(\vec{z}_A,\vec{z}_{B_1},\ldots,\vec{z}_{B_n})\in\R^{2(N_A+N_{B_1}+\cdots+N_{B_n})}$ since
$$
\langle R(\vec{z})\rangle_\varrho=\vec{d}_{\varrho_A}^T\vec{S}_A\vec{z}_A+\vec{d}_{\varrho_{B_1}}^T\vec{S}_{B_1}\vec{z}_{B_1}+\cdots+\vec{d}_{\varrho_{B_n}}^T\vec{S}_{B_n}\vec{z}_{B_n}.
$$
We may assume in this proof that $\vec{d}_{\varrho_i}=0$ for all $i=1,\ldots,\,n$. Indeed, if this is not the case and not all of these displacement vectors $\vec{d}_{\varrho_i}=:\vec{d}_i=(\vec{d}_A,\vec{d}_{B_i})$ are zero, the states $\tilde{\varrho}_i=W_{AB_i}(\vec{d}_i)\rho_iW_{AB_i}(\vec{d}_i)^*$ are easily found to have null displacement vectors; the operators $W_C(\vec{z})$ above are the displacement operators (or Weyl operators), $W_C(\vec{z})=e^{iR_C(\vec{z})}$ for all $C\subseteq\{A,B_1,\ldots,\,B_n\}$ and $\vec{z}\in\R^{N_C}$. It follows that the state $\tilde{\varrho}=W(\vec{d}_A,\vec{d}_{B_1},\ldots,\vec{d}_{B_n})\varrho W(\vec{d}_A,\vec{d}_{B_1},\ldots,\vec{d}_{B_n})^*$ is such that $\tilde{\varrho}_{AB_i}=\tilde{\varrho}_i$ for all $i=1,\ldots,\,n$ and that the displacement vector of $\tilde{\varrho}$ vanishes. Let us thus assume that the displacement vectors of $\varrho_i$ vanish and $\varrho$ is a state over the entire system $A\&B_1\&\cdots\&B_n$, whose displacement vector vanishes, such that $\varrho_{AB_i}=\varrho_i$ for all $i=1,\ldots,\,n$.

We next move on to the second moments and show that $\langle R(\vec{z})^2\rangle_\varrho$ is defined for all $\vec{z}\in\R^{2(N_A+N_{B_1}+\cdots+N_{N_n})}$. Pick $\vec{z}=(\vec{z}_A,\vec{z}_{B_1},\ldots,\vec{z}_{B_n})\in\R^{2(N_A+N_{B_1}+\cdots+N_{B_n})}$ where $\vec{z}_A\in\R^{2N_A}$ and $\vec{z}_{B_i}\in\R^{2N_{B_i}}$ for $i=1,\ldots,\,n$. Squaring the operator $R(\vec{z})$ and (na\"{i}vely) calculating the expectation value of this square in $\varrho$, we easily see that
\begin{eqnarray*}
\langle R(\vec{z})^2\rangle_\varrho&=&\langle R_A(\vec{z}_A)^2\rangle_{\varrho_A}+2\sum_{i=1}^n\langle R_A(\vec{z}_A)\otimes R_{B_i}(\vec{z}_{B_i})\rangle_{\varrho_{AB_i}}\\
&&+2\sum_{i,j:\,i<j}\langle R_{B_i}(\vec{z}_{B_i})\otimes R_{B_j}(\vec{z}_{B_j})\rangle_{\varrho_{B_iB_j}}+\sum_{i=1}^n\langle R_{B_i}(\vec{z}_{B_i})^2\rangle_{\varrho_{B_i}}.
\end{eqnarray*}
According to the preceding lemma, the terms in the second sum above are defined and, as the states $\rho_A=(\rho_i)_A$, $\rho_{AB_i}=\rho_i$, and $\rho_{B_i}=(\rho_i)_{B_i}$, $i=1,\ldots,\,n$, are Gaussian, the other terms are defined as well; recall that the reduced states of Gaussian states are Gaussian as well. Thus, we have $\langle R(\vec{z})^2\rangle_\varrho<\infty$.

We are now ready to employ the quantum central limit theorem to derive the desired Gaussian state $\varrho^G$ from $\varrho$. Denote by $\rho^G$ the Gaussian state with the covariance matrix $\vec{C}_\varrho$ (and null displacement vector). For any $n\in\N$, let $\varrho^{(n)}$ be the state over $A\&B_1\&\cdots\&B_n$ whose characteristic function is given by
$$
\tr[\varrho^{(n)}W(\vec{z})]=\chi_{\varrho^{(n)}}(\vec{z})=\chi_\varrho(\vec{z}/\sqrt{n})^n
$$
for all $\vec{z}\in\R^{2(N_A+N_{B_1}+\cdots+N_{B_n})}$. According to the quantum central limit theorem presented in \cite{CuHu1971}, these characteristic functions converge pointwise to $\chi_{\varrho^G}$ as $n\to\infty$ and, since the Weyl operators span an ultraweakly dense operator system of the algebra of bounded operators over our system, it follows that the sequence $(\varrho^{(n)})_{n=1}^\infty$ converges to $\varrho^G$ $\sigma$-weakly, i.e.,\ for all bounded linear operators $D$ over the system, $\lim_{n\to\infty}\tr[\varrho^{(n)}D]=\tr[\varrho^G D]$. Since the evaluation functions
$$
\sigma\mapsto\tr[\sigma(D_A\otimes\id_{B_1}\otimes\cdots\otimes\id_{B_{i-1}}\otimes D_{B_i}\otimes\id_{B_{i+1}}\otimes\cdots\otimes\id_{B_n}]
$$
for all bounded operators $D_A$ over system $A$ and $D_{B_i}$ over system $B_i$, $i=1,\ldots,\,n$, are $\sigma$-weakly continuous, it immediately follows that the set of those states of the entire system $A\&B_1\&\cdots\&B_n$ with the fixed bipartite margins $\varrho_1,\ldots,\,\varrho_n$ is $\sigma$-weakly closed (as the intersection of preimages of the above evaluation functions) and thus $\varrho^G$ also has the required margins since $\varrho^{(n)}$ has the bipartite margins $\varrho_i$, $i=1,\ldots,\,n$, for all $n\in\N$ as can be easily seen (to convince oneself of this, see the following calculation). However, let us construct a more direct proof for this. We have, for any $i\in\{1,\ldots,n\}$, $\vec{z}_A\in\R^{2N_A}$, and $\vec{z}_{B_i}\in\R^{2N_{B_i}}$,
\begin{align*}
&\tr[\varrho^G_{AB_i}W_{AB_i}(\vec{z}_A,\vec{z}_{B_i})]\\
=&\tr[\varrho^G\big(W_A(\vec{z}_A)\otimes\id_{B_1}\otimes\cdots\otimes\id_{B_{i-1}}\otimes W_{B_i}(\vec{z}_{B_i})\otimes\id_{B_{i+1}}\otimes\cdots\otimes\id_{B_n}\big)]\\
=&\tr[\varrho^G W(\vec{z}_A,0,\ldots,0,\vec{z}_{B_i},0,\ldots,0)]=\chi_{\varrho^G}(\vec{z}_A,0,\ldots,0,\vec{z}_{B_i},0,\ldots,0)\\
=&\lim_{n\to\infty}\chi_{\varrho^{(n)}}(\vec{z}_A,0,\ldots,0,\vec{z}_{B_i},0,\ldots,0)=\lim_{n\to\infty}\chi_{\varrho}(n^{-1/2}\vec{z}_A,0,\ldots,0,n^{-1/2}\vec{z}_{B_i},0,\ldots,0)^n\\
=&\lim_{n\to\infty}\tr[\varrho W(n^{-1/2}\vec{z}_A,0,\ldots,0,n^{-1/2}\vec{z}_{B_i},0,\ldots,0)]^n\\
=&\lim_{n\to\infty}\tr[\varrho_{AB_i}W_{AB_i}(n^{-1/2}\vec{z}_A,n^{-1/2}\vec{z}_{B_i})]^n=\lim_{n\to\infty}\tr[\varrho_i W_{AB_i}(n^{-1/2}\vec{z}_A,n^{-1/2}\vec{z}_{B_i})]^n\\
=&\lim_{n\to\infty}\chi_{\varrho_i}(n^{-1/2}\vec{z}_A,n^{-1/2}\vec{z}_{B_i})^n=\lim_{n\to\infty}\chi_{\varrho_i}(\vec{z}_A,\vec{z}_{B_i})=\chi_{\varrho_i}(\vec{z}_A,\vec{z}_{B_i})=\tr[\varrho_i W_{AB_i}(\vec{z}_A,\vec{z}_{B_i})],
\end{align*}
where we have used the fact that $G(n^{-1/2}\vec{z}_A,n^{-1/2}\vec{z}_{B_i})^n=G(\vec{z}_A,\vec{z}_{B_i})$ for a Gaussian distribution $G$ with a vanishing mean in the third-to-last equality above. Since the Weyl operators in the subsystem $A\&B_i$ span an ultraweakly dense operator system, we have $\varrho^G_{AB_i}=\varrho_i$. This concludes the proof for the existence of the state $\varrho^G$ of the claim.

We finally prove the last claim regarding the covariance matrices of $\varrho_i$ and the condition for the solvability of the marginal problem they set up in Inequality \ref{eq:GaussianMargCond}. Let $\vec{C}_{\varrho_i}$ be as in the claim for all $i=1,\ldots,\,n$ and suppose $\varrho^G$ is the Gaussian solution for the marginal problem. It follows that, for each $i,\,j\in\{1,\ldots,n\}$, $i<j$, there is a real $(2N_{B_i}\times 2N_{B_j})$-matrix $\vec{Y}_{i,j}$ such that
$$
\vec{C}_{\varrho^G}=\left(\begin{array}{ccccc}
    \vec{C}_A &\vec{X}_1 &\vec{X}_2 &\cdots &\vec{X}_n  \\
    \vec{X}_1^T &\vec{C}_{B_1} &\vec{Y}_{1,2} &\cdots &\vec{Y}_{1,n}  \\
    \vec{X}_2^T &\vec{Y}_{1,2}^T &\vec{C}_{B_2} &\cdots &\vec{Y}_{2,n}  \\
    \vdots &\vdots &\vdots &\ddots &\vdots  \\
    \vec{X}_n^T &\vec{Y}_{1,n}^T &\vec{Y}_{2,n}^T &\cdots &\vec{C}_{B_n}
\end{array}
\right);
$$
this follows in a straight-forward manner by checking the covariance matrices of the $A\& B_i$-marginals of $\varrho^G$. This is an allowed covariance matrix if and only if $\vec{C}_{\varrho^G}+i\vec{S}_A\oplus\vec{S}_{B_1}\oplus\cdots\oplus\vec{S}_{B_n}\geq0$ which is equivalent to Equation (\ref{eq:GaussianMargCond}).
\end{proof}

Let us go on to studying Gaussian channels. Recall that matrices $\vec{K}$ and $\vec{L}$ and a vector $\vec{m}$ define a Gaussian channel $\Phi_{A\to B}=\Phi_{\vec{K},\vec{L},\vec{m}}$ if and only if
\begin{equation}\label{eq:GaussianChCond}
\vec{K}+i\vec{S}_B-i\vec{L}^T\vec{S}_A\vec{L}\geq0.
\end{equation}
We also recall that the Gaussian channel $\Phi_{\vec{K},\vec{L},\vec{m}}$ is also characterized by the Weyl operator condition for the Heisenberg dual:
\begin{equation}\label{eq:GaussianWeyl}
\Phi_{\vec{K},\vec{L},\vec{m}}^*\big(W_B(\vec{z})\big)=e^{-\frac{1}{4}\vec{z}^T\vec{K}\vec{z}+i\vec{m}^T\vec{z}}W_A(\vec{L}\vec{z}),\qquad\vec{z}\in\R^{2N_B}.
\end{equation}

According to \cite{HoWe2001}, for any Gaussian state $\varrho_A$, there is Gaussian purification, i.e.,\ a purifying vector $\ket{\Omega_A}$ such that $\dyad{\Omega_A}$ is a Gaussian state. Since the identity channel is trivially Gaussian and the tensor product of two Gaussian channels is Gaussian (as can be easily verified), it follows that $({\rm id}_A\otimes\Phi_{A\to B})(\dyad{\Omega_A})$ is Gaussian. Especially, when $\varrho_A$ above is faithful, we obtain a state-channel dualism that maps Gaussian channels into Gaussian states.

Let us fix $\alpha>1$ and denote $\lambda:=\ln{\left(\frac{\alpha+1}{\alpha-1}\right)}$. Define the Hamiltonian of the Harmonic oscillator on system $A$ as $H_{\rm osc}=\frac{1}{2}(\|\vec{Q}\|^2+\|\vec{P}\|^2)$ where $\vec{Q}$ and $\vec{P}$ are the position and momentum operators of system $A$ in vector form. We fix the faithful (sic) state $\varrho_A:=2\sinh{(\lambda/2)}\,e^{-\lambda H_{\rm osc}}$ the characteristic function of which is given by $\tr[\varrho_A W_A(\vec{z})]=\chi_{\varrho_A}(\vec{z})=e^{-\frac{1}{4}\alpha\|\vec{z}\|^2}$ for all $\vec{z}\in\R^{2N_A}$ \cite{GaZo}. This means that $\varrho_A$ is Gaussian and is defined by a vanishing displacement vector and the covariance matrix $\alpha\id_{2N_A}$. Denote by $\vec{Z}$ the block-diagonal $(2N_A\times 2N_A)$-matrix with the $(2\times 2)$-blocks
$$
\sigma_z:=\left(\begin{array}{cc}
1&0\\
0&-1
\end{array}\right)
$$
on the diagonal. Following \cite{HoWe2001}, we may give $\varrho_A$ the purification $\ket{\Omega_A}\in\hil_A\otimes\hil_A$ such that $\dyad{\Omega_A}$ is associated with the covariance matrix
\begin{equation}\label{eq:C0}
\vec{C}_0:=\left(\begin{array}{cc}
\alpha\id_{2N_A}&\sqrt{\alpha^2-1}\vec{Z}\\
\sqrt{\alpha^2-1}\vec{Z}&\alpha\id_{2N_A}
\end{array}\right).
\end{equation}
Indeed, one may easily verify that $\big((\vec{S}_A\otimes\vec{S}_A)^T\vec{C}_0\big)^2=-\id_{4N_A}$ implying that $\vec{C}_0$ corresponds to a pure state. We keep the above state $\varrho_A$, the purification $\ket{\Omega_A}$, and the matrix $\vec{C}_0$ fixed in the proof of the following result which generalizes Observation \ref{obs:GaussianComp}.

\begin{theorem}\label{theor:GaussComp}
Gaussian channels $\Phi_i=\Phi_{\vec{K}_i,\vec{L}_i,\vec{m}_i}$ (with input system $A$ and output system $B_i$), $i=1,\ldots,\,n$, are compatible if and only if, for any $i,\,j\in\{1,\ldots,n\}$, $i<j$, there is a real $(2N_{B_i}\times 2N_{B_j})$-matrix $\vec{Z}_{i,j}$ such that
\begin{equation}\label{eq:CompMatrixCond}
\left(\begin{array}{cccc}
\vec{K}_1+i\vec{S}_{B_1}-i\vec{L}_1^T\vec{S}_A\vec{L}_1&\vec{Z}_{1,2}-i\vec{L}_1^T\vec{S}_A\vec{L}_2&\cdots&\vec{Z}_{1,n}-i\vec{L}_1^T\vec{S}_A\vec{L}_n\\
\vec{Z}_{1,2}^T+i\vec{L}_2^T\vec{S}_A^T\vec{L}_1&\vec{K}_2+i\vec{S}_{B_2}-i\vec{L}_2^T\vec{S}_A\vec{L}_2&\cdots&\vec{Z}_{2,n}-i\vec{L}_2^T\vec{S}_A\vec{L}_n\\
\vdots&\vdots&\ddots&\vdots\\
\vec{Z}_{1,n}^T+i\vec{L}_n^T\vec{S}_A^T\vec{L}_1&\vec{Z}_{2,n}^T+i\vec{L}_n^T\vec{S}_A^T\vec{L}_2&\cdots&\vec{K}_n+i\vec{S}_{B_n}-i\vec{L}_n^T\vec{S}_A\vec{L}_n
\end{array}\right)\geq0.
\end{equation}
Moreover, if $\Phi_i$, $i=1,\ldots,\,n$, are compatible, they have a Gaussian joint channel.
\end{theorem}

\begin{proof}
Denote $S_i:=S_{\ket{\Omega_A}}(\Phi_i)$ for $i=1,\ldots,\,n$. For any $i=1,\ldots,\,n$, we have ${\rm id}_{A}\otimes\Phi_i=\Phi_{\vec{K}'_i,\vec{L}'_i,\vec{m}'_i}$ where $\vec{K}'_i:=0\oplus \vec{K}_i$, $\vec{L}'_i:=\id_{2N_A}\oplus \vec{L}_i$, and $\vec{m}'_i:=0\oplus\vec{m}_i$. It follows that the covariance matrix of $S_i$ is
\begin{equation}\label{eq:Ci}
C_i:=(\vec{L}'_i)^T\vec{C}_0\vec{L}'_i+\vec{K}'_i=\left(\begin{array}{cc}
\alpha\id_{2N_A}&\sqrt{\alpha^2-1}\,\vec{Z}\vec{L}_i\\
\sqrt{\alpha^2-1}\,\vec{L}_i^T\vec{Z}&\alpha \vec{L}_i^T\vec{L}_i+\vec{K}_i
\end{array}\right)
\end{equation}
for all $i=1,\ldots,\,n$. The channels $\Phi_i$, $i=1,\ldots,\,n$ are compatible if and only if the marginal problem involving the states $S_i$, $i=1,\ldots,\,n$, has a solution. This is equivalent, according to Theorem \ref{theor:GaussianMarg}, to the existence, for all $i,\,j\in\{1,\ldots,\,n\}$, $i<j$, of a real $(2N_{Bi}\times 2N_{B_j})$-matrix $Y_{i,j}$ such that
\begin{equation}\label{eq:apu1}
\left(\begin{array}{cc}
\alpha\id_{2N_A}+i\vec{S}_A&\vec{X}\\
\vec{X}^T&\vec{Y}
\end{array}\right)\geq0
\end{equation}
where we have denoted by $\vec{X}$ the horizontal row $\sqrt{\alpha^2-1}(\vec{Z}\vec{L}_1\,\cdots\,\vec{Z}\vec{L}_n)$ of blocks and
$$
\vec{Y}:=\left(\begin{array}{cccc}
\alpha \vec{L}_1^T\vec{L}_1+\vec{K}_1+i\vec{S}_{B_1}&\vec{Y}_{1,2}&\cdots&\vec{Y}_{1,n}\\
\vec{Y}_{1,2}^T&\alpha \vec{L}_2^T\vec{L}_2+\vec{K}_2+i\vec{S}_{B_2}&\cdots&\vec{Y}_{1,n}\\
\vdots&\vdots&\ddots&\vdots\\
\vec{Y}_{1,n}^T&\vec{Y}_{2,n}^T&\cdots&\alpha \vec{L}_n^T\vec{L}_n+\vec{K}_n+i\vec{S}_{B_n}
\end{array}\right).
$$
Since the matrix $\alpha\id_{2N_A}+i\vec{S}_A$ is positive and invertible (guaranteed by $\alpha>1$), this is equivalent with the Schur complement of $\alpha\id_{2N_A}+i\Omega_A$ in the matrix appearing on the LHS of Inequality \eqref{eq:apu1} being positive, i.e.,\ with $\vec{Y}-\vec{X}^T(\alpha\id_{2N_A}+i\vec{S}_A)^{-1}\vec{X}\geq0$. Using the easily verifiable fact that
$$
\vec{Z}(\alpha\id_{2N_A}+i\vec{S}_A)^{-1}\vec{Z}=(\alpha\id_{2N_A}-i\vec{S}_A)^{-1}=\frac{1}{\alpha^2-1}(\alpha\id_{2N_A}+i\vec{S}_A),
$$
one immediately obtains $\vec{X}^T(\alpha\id_{2N_A}+i\vec{S}_A)^{-1}\vec{X}=\big(\vec{L}_i^T(\alpha\id_{2N_A}+i\vec{S}_A)\vec{L}_j\big)_{i,j=1}^n$. Defining, for all $i,\,j\in\{1,\ldots,\,n\}$, $i<j$, the real $(2N_{B_i}\times 2N_{B_j})$-matrix $\vec{Z}_{i,j}:=\vec{Y}_{i,j}-\alpha \vec{L}_i^T\vec{L}_j$, the inequality $\vec{Y}-\vec{X}^T(\alpha\id_{2N_A}+i\vec{S}_A)^{-1}\vec{X}\geq0$ becomes Inequality \eqref{eq:CompMatrixCond}. We immediately obtain the necessary and sufficient condition for the compatibility of $\Phi_i$, $i=1,\ldots,\,n$, as stated in the claim.

Assume now that $\Phi_i$, $i=1,\ldots,\,n$, are compatible, i.e.,\ for all $i,\,j\in\{1,\ldots,n\}$, $i<j$, there is a real $(2N_{B_i}\times 2N_{B_j})$-matrix $\vec{Z}_{i,j}$ such that Inequality \eqref{eq:CompMatrixCond} is satisfied. Define the Gaussian channel $\Psi:=\Phi_{\vec{K},\vec{L},\vec{m}}$ (with the input system $A$ and output system $B_1\&\cdots\&B_n$) by setting
$$
\vec{K}:=\left(\begin{array}{cccc}
\vec{K}_1&\vec{Z}_{1,2}&\cdots&\vec{Z}_{1,n}\\
\vec{Z}_{1,2}^T&\vec{K}_2&\cdots&\vec{Z}_{2,n}\\
\vdots&\vdots&\ddots&\vdots\\
\vec{Z}_{1,n}^T&\vec{Z}_{2,n}^T&\cdots&\vec{K}_n
\end{array}\right),
$$
$\vec{L}:=(\vec{L}_1\,\cdots\,\vec{L}_n)$, and $\vec{m}:=\vec{m}_1\oplus\cdots\oplus\vec{m}_n$. This is, indeed, an allowed Gaussian channel as the inequality
$$
\vec{K}+i\vec{S}_{B_1}\oplus\cdots\oplus\vec{S}_{B_n}-i\vec{L}^T\vec{S}_A\vec{L}\geq0
$$
that the matrices $\vec{K}$ and $\vec{L}$ must satisfy is easily found to be the same as Inequality \eqref{eq:CompMatrixCond}. It is immediate that $\Psi_{A\to B_i}$ obtained from $\Psi$ by tracing out other output systems except for $B_i$ coincides with $\Phi_i$ for all $i=1,\ldots,\,n$ since, for any $i=1,\ldots,\,n$, $\vec{z}\in\R^{2N_{B_i}}$, and denoting by $\vec{w}\in\R^{2(N_{B_1}+\cdots+N_{B_n})}$ the vector $(0_1,\ldots,0_{i-1},\vec{z}_i,0_{i+1},\ldots,0_n)$, where $0_j$ is the zero-vector in $\R^{2N_{B_j}}$, $j=1,\ldots,\,n$, and by $W_B$ the Weyl representation associated with the subsystem $B_1\&\cdots\&B_n$,
\begin{eqnarray*}
\Psi_{A\to B_i}^*\big(W_{B_i}(\vec{z})\big)&=&\Psi^*\big(\id_{B_1}\otimes\cdots\otimes\id_{B_{i-1}}\otimes W_{B_i}(\vec{z})\otimes\id_{B_{i+1}}\otimes\cdots\otimes\id_{B_n}\big)\\
&=&\Psi^*\big(W_B(\vec{w})\big)=e^{-\frac{1}{4}\vec{w}^T\vec{K}\vec{w}+i\vec{m}^T\vec{w}}W_A(\vec{L}\vec{w})\\
&=&e^{-\frac{1}{4}\vec{z}^T\vec{K}_i\vec{z}+\vec{m}_i^T\vec{z}}W_A(\vec{L}_i\vec{z})=\Phi_i^*\big(W_{B_i}(\vec{z})\big)
\end{eqnarray*}
where we have used Equation \eqref{eq:GaussianWeyl}.
\end{proof}

\end{document}